\documentclass[final,twoside,11pt]{entics} 
\usepackage{enticsmacro}
\usepackage{amsmath,amssymb}
\usepackage{tikz-cd}
\def\conf{ISDT 2022} 	
\volume{2}			
\def\volu{2}			
\def\papno{7}			
\def\mydoi{{\normalshape  \href{https://doi.org/10.46298/entics.10362}{10.46298/entics.10362}}}
\def\copynames{Z. Lyu, X, Xie, H. Kuo} 
\def\runtitle{A note on the category of c-spaces}

\def\lastname{Lyu, Xie and Kou}
\begin{document}
\begin{frontmatter}
  \title{A Note on the Category of c-spaces} 
  \author{Zhenchao Lyu\thanksref{a}\thanksref{ALL}\thanksref{2}}	
   \author{Xiaolin Xie\thanksref{a}\thanksref{ALL}\thanksref{3}}
   \author{Hui Kou\thanksref{a}\thanksref{ALL}\thanksref{4}}
   \address[a]{Department of Mathematics\\ Sichuan University\\				
    Chengdu, 610064, China}  							
  \thanks[ALL]{Supported by the NSF of China (Nos. 11871353, 12001385) and ``the Fundamental Research Funds for the Central Universities" No. 2021SCU12108}   
   \thanks[2]{Email: \href{zhenchaolyu@scu.edu.cn} {\texttt{\normalshape
       zhenchaolyu@scu.edu.cn}}} 
   
  \thanks[3]{Email:  \href{xxldannyboy@163.com} {\texttt{\normalshape
        xxldannyboy@163.com}}}
  \thanks[4]{Email:  \href{kouhui@scu.edu.cn} {\texttt{\normalshape
   			kouhui@scu.edu.cn}}}
    
\begin{abstract} 
 We prove that the category of c-spaces with continuous maps is not cartesian closed. As a corollary, it follows that the category of locally finitary compact spaces with continuous maps also is  not cartesian closed.\\

\end{abstract}
\end{frontmatter}
\section{Introduction}\label{intro}

Many people have been trying to extend domain theory to general topological spaces, see \cite{ershov73,Erne91,Ershov93,Goubault13,Battenfeld07}. Directed spaces are introduced by Kou and Yu independently \cite{KouY15} in 2014 for generalizing the concept of Scott spaces, which is equivalent to that of $T_{0}$ monotone determined spaces introduced by Ern\'e \cite{Erne09}. In the same paper Kou and Yu proved that the category of directed spaces with continuous maps (\textbf{DTop} for short) is cartesian closed. There are many important directed spaces in domain theory, for instance locally finitary compact spaces are directed spaces; in particular  c-spaces and Alexandroff spaces are directed spaces. Since the category of continuous domains is not cartesian closed, and since the position of the category of c-spaces in the category of directed spaces is similar to that of continuous domains in dcpos \cite{Chen22}, a natural question arises: Is the category of c-spaces cartesian closed? In this short note, we  answer this question in the negative.


\section{Preliminaries}
We refer to~\cite{Gierz03,JungA94,Goubault13} for the standard definitions and notations in order theory, topology and domain theory. A partially ordered set $D$ is called a dcpo if every directed subset of $D$ has a supremum in $D$. A upper set $U$ is called a Scott open set if for any directed set $A\subseteq D$, $\bigvee A\in U$ implies $A$ intersects $U$.

For a topological space $X$,  we use $\mathcal O(X)$ to denote the lattice of open subsets of~$X$. We require that all topological spaces are $T_{0}$ in this note. Let $X$ be a $T_{0}$ space, the specializing order $\leq$ is defined as follows $\colon$ $x\leq y$ if $x$ belongs to the closure of point $y$. 
A topological space is a \emph{c-space} if  for any $x\in X$ and any open neighbourhood $U$ of $x$, there is a point $y\in U$ such that $x\in \text{int}\,({\uparrow}y)$. A space $X$ is locally finitary compact if for any $x\in X$ and its open neighborhood $U$, there is a finite subset $F$ of $U$ such that $x\in \text{int}\,({\uparrow}F)$. 

Let $X$ be a $T_{0}$ space and $\leq$  the specialization order over $X$.  A topological space $X$ is called a Scott space if $(X,\leq)$ is a dcpo and the topology on $X$ is equal to the Scott topology on $(X,\leq)$.  
Every directed set $D$ of $X$ under specialization order can be regarded as a monotone net, we say $D$ converges to $x$ iff for every open neighborhood $U$ of $x$, $D\bigcap U\neq\emptyset$.
We say that $V$ is a directed open set of $X$ if for all  directed set $D$ which converges to some point of $V$, then $D\bigcap V\neq\emptyset$. It is easy to see that every directed open set is an upper set.

\begin{definition}\cite{KouY15}
 Let $X$ be a $T_0$ space. If every directed open set of $X$ is also an open set, then we say that $X$ is a directed space. 
\end{definition}

There are many important spaces in domain theory which also are directed spaces.
\begin{example}	
	\begin{enumerate}
		\item Every poset with Scott topology is a directed space.
		\item All c-spaces are directed spaces. In particular, every Alexandroff space is a directed space.
		\item All locally finitary compact spaces are directed spaces. By the way, every c-space is locally finitary compact.
	\end{enumerate}
\end{example}

Next we introduce the concept of the exponential object in general category.
\begin{definition}
	Given two objects $X,Y$ in a category $\mathcal{C}$ with binary products, an exponential object, if it exists, is an object $Y^{X}$ with a morphism $App\colon Y^{X}\times X\rightarrow Y$ such that for every morphism $f\colon Z\times X\rightarrow Y$, there is a unique morphism $\bar{f}\colon Z\rightarrow Y^{X}$ such that the following diagram commutes:

	\begin{equation*}
		\begin{tikzcd}[row sep=huge, column sep = huge]
			Z\times X  \arrow[dr, "f"'] \arrow[r, "\bar{f}\times id_{X}"] & Y^{X}\times X \arrow[d, "App"] \\
			& Y
		\end{tikzcd}
	\end{equation*}
	
\end{definition}

The following result describes the underlying set of the exponential object in \textbf{Top}.
\begin{proposition}\rm\cite{Goubault13}
	Let $\mathcal{C}$ be any full subcategory of \rm\textbf{Top} with finite products, and assume that $1=\{\star\}$ is an object of $\mathcal{C}$. Let $X,Y$ be two objects of $\mathcal{C}$ that have an exponential object $Y^{X}$ in $\mathcal{C}$.
	
	Then there is a unique homeomorphism $\theta\colon Y^{X}\rightarrow [X\rightarrow Y]$, for some unique topology on $[X\rightarrow Y]$ (the set of all continuous functions from $X$ to $Y$) , such that $App(h,x)=\theta(h)(x)$ for all $h\in Y^{X}, x\in X$.
	
	Moreover, $\bar{f}(z)$ is the image by $\theta^{-1}$ of $f(z,\_)$ for all $f\colon Z\times X\rightarrow Y, z\in Z$. 
\end{proposition} 	

Remark: By the above result, we always let the exponential object in $\mathcal{C}$ be the set $[X\rightarrow Y]$ with some unique topology if it exists.

\begin{theorem}\rm\cite{KouY15}
	The category of directed spaces and continuous maps is cartesian closed.
\end{theorem}

Next, we build a relationship between directed spaces and Scott spaces, which will be used later.

\begin{definition}
	Let $X$ be a $T_0$ space. If $X$ with the specialization order is a dcpo and every open set of $X$ is Scott open in $(X,\leq)$. Then we say that $X$ is a d-space.
\end{definition}

\begin{lemma}
	A directed space is a Scott space iff it is a d-space.
\end{lemma}
\begin{proof}
	We only need to show the ``if" part. Let $X$ be a d-space and a directed space, obviously every  open set of $X$ is Scott open of $(X,\leq)$ since $X$ is a d-space. Now take any Scott open set $U$ of $(X,\leq)$ and for any directed set $D$ converges to some point $x$ of $U$. Assume that $D\bigcap U=\emptyset$, then $b=\bigvee^{\uparrow} D\not\in U$. It follows that $x\in X\backslash{\downarrow}b$.  Because $D$ converges to $x$ and $X\backslash{\downarrow}b$ is open in X, there is some $d\in D$ such that  $d\in X\backslash{\downarrow}b$,, a contradiction. Hence the assumption is wrong. It means that $U$ is a directed open set of $X$.  Since $X$ is a directed space, the topology on $X$ is exactly the Scott topology on $(X,\leq)$.
\end{proof}

We list some results about separate continuity and joint continuity.
\begin{theorem}\rm\cite{Lawson85}
	Let $E$ be a $T_{0}$ space. The following conditions are equivalent:
	\begin{enumerate}
		\item $E$ is locally finitary compact.
		\item For all $T_{0}$ space X, if a map from $X\times E$ is separately continuous, then it is jointly continuous.

	\end{enumerate}
	
\end{theorem}

\begin{corollary}\label{sepa}
	Let $X$ be a c-space and $Y$  a $T_{0}$ space. For any $T_{0}$ space $Z$, a map $f\colon X\times Y\rightarrow Z$ is continuous (i.e. jointly continuous) iff it is separately continuous. 
\end{corollary}


\section{The category of c-spaces}

We now prove  our main result.

\begin{theorem}
	The category of c-spaces with continuous maps (\textbf{CS} for short) is not cartesian closed.	
\end{theorem}

\begin{proof}
	Let $\mathbb{Z}^{-}$ be the set of non-positive integers with the Scott topology.
	Assume $\textbf{CS}$ is a ccc. It is easy to see that the topological product $X\times Y$ is the categorical product because $X\times Y$ is a c-space.  Since $\textbf{CS}$ is cartesian closed, there exists  exponential topology $\tau$ on $[\mathbb{Z}^{-}\rightarrow\mathbb{Z}^{-}]$, we denote by $[\mathbb{Z}^{-}\rightarrow\mathbb{Z}^{-}]_{\tau}$. Then for any c-space $Y$ and any  map $f\colon Y\times \mathbb{Z}^{-}\rightarrow \mathbb{Z}^{-}$,  $f$ is continuous iff $\bar{f}\colon Y\rightarrow [\mathbb{Z}^{-}\rightarrow\mathbb{Z}^{-}]_{\tau}$ is continuous. 
	
	\textbf{Claim $1$:} The specialization order on $[\mathbb{Z}^{-}\rightarrow\mathbb{Z}^{-}]_{\tau}$ is equal to the pointwise order.
	
	For any $g_1, g_2\in [\mathbb{Z}^{-}\rightarrow\mathbb{Z}^{-}]_{\tau}$ with $g_1\leq_{\tau}g_2 ~(g_1\neq g_2)$, take $Y=\mathbb{S}$ with Scott topology. A map $\theta\colon \mathbb{S}\rightarrow [\mathbb{Z}^{-}\rightarrow\mathbb{Z}^{-}]_{\tau}$ is defined as $\theta(1)=g_2, \theta(0)=g_1$.
	Then $\theta$ is continuous. Hence $\hat{\theta}:\mathbb{S}\times \mathbb{Z}^{-}\rightarrow\mathbb{Z}^{-}$ is continuous. It follows that $$g_1(x)=\hat{\theta}(0,x)\leq \hat{\theta}(1,x)=g_2(x)$$ for any $x\in X$.
	
	For any $g_1, g_2\in [\mathbb{Z}^{-}\rightarrow\mathbb{Z}^{-}]_{\tau}$ with $g_1\leq g_2$, consider a continuous map $f\colon\mathbb{S}\times \mathbb{Z}^{-}\rightarrow\mathbb{Z}^{-}$ which is defined as $f(0,x)=g_1(x), f(1,x)=g_2(x) ~\forall x\in X$. It follows that the transpose
	map $\bar{f}$ is continuous hence monotone, which implies that $$g_1=\bar{f}(0)\leq_{\tau}\bar{f}(1)=g_2.$$
	
	\textbf{Claim $2$:} $[\mathbb{Z}^{-}\rightarrow\mathbb{Z}^{-}]_{\tau}$ is a d-space.
	
	We only need to show that every directed subfamily $(g_i)_{i\in I}$ of $[\mathbb{Z}^{-}\rightarrow\mathbb{Z}^{-}]_{\tau}$ converges to its supremum $g=\bigvee_{i\in I}^{\uparrow}g_i$. Let $Y$ be the set $I\cup\{\infty\}$ with the topology generated by
	$\{{\uparrow}i\cup\{\infty\}:i\in I\}$. Obviously $Y$ is a c-space. Consider a map $f\colon Y\times \mathbb{Z}^{-}\rightarrow\mathbb{Z}^{-}$ which is defined as $f(\infty,x)=g(x), f(i,x)=g_i(x)$.
	It is easy to see that $f$ is continuous since $f$ is continuous iff it is separately continuous by \ref{sepa}. It follows that $\bar{f}\colon Y\rightarrow [\mathbb{Z}^{-}\rightarrow\mathbb{Z}^{-}]_{\tau}$ is continuous, and so $(g_i=\bar{f}(i))_i$ converges to $\bar{f}(\infty)=g$.
	
	Therefore $\tau$ is just the Scott topology on $[\mathbb{Z}^{-}\rightarrow\mathbb{Z}^{-}]$.
	But from \cite{Jung89} we know that $[\mathbb{Z}^{-}\rightarrow\mathbb{Z}^{-}]$ is not a continuous domain, hence it is not a c-space, a contradiction.
\end{proof}

\begin{theorem}\rm\cite{Gierz03}
	A meet continuous dcpo is a continuous dcpo iff it is a quasicontinuous dcpo.
\end{theorem}

Notice that $[\mathbb{Z}^{-}\rightarrow\mathbb{Z}^{-}]$ is a meet continuous semilattice which is not continuous, hence it is not a quasicontinuous dcpo. Then we have the following result.
\begin{corollary}
	The category of locally finitary compact spaces with continuous maps is not cartesian closed.
\end{corollary}	

\section*{Acknowledgement}
We are grateful to the anonymous referees for useful comments and suggestions. We also would like to thank Yuxu Chen for helpful discussions.

\bibliographystyle{entics}
\bibliography{note}
  
  
  
  
  

\end{document}